\documentclass{article}

\usepackage{amsmath,amsfonts,amssymb,enumerate,hanmath}
\usepackage{bm,upgreek}
\usepackage{natbib}

\usepackage{graphicx}

\newcommand{\tens}[1]{\mathsf{#1}}
\renewcommand{\vec}[1]{\mathbf{#1}}
\newcommand{\mat}[1]{\mathbf{#1}}

\newtheorem{remark}[theorem]{Remark.}

\newcommand\qedsymbol{\hbox{\rlap{$\sqcap$}$\sqcup$}}
\newenvironment{proof}{\textbf{Proof.}}{\\ \mbox{ }\hfill\qedsymbol\\}
\newenvironment{prooftheorem}{\textbf{Proof of the theorem.}}{\hfill\qedsymbol}

\theoremstyle{plain}
\newtheorem{theoremcited}[theorem]{Theorem}

\newtheorem{corollarycited}[theorem]{Corollary}

\theoremstyle{plain}

\renewcommand{\im}{\mathop{\mathrm{Im}}}
\renewcommand{\re}{\mathop{\mathrm{Re}}}

\newcommand{\D}{\mathrm{D}}

\newcommand{\LICM}{\mathfrak{L}}

\newcommand{\BF}{\mathfrak{B}}

\newcommand{\CBF}{\mathfrak{F}}
\renewcommand{\Sigma}{\mathfrak{S}}

\newcommand{\upi}{\uppi}

\author{Ma{\l}gorzata Seredy\'{n}ska\\
Institute of Fundamental Technological Research\\
Polish Academy of Sciences\\
ul. Pawi\'{n}skiego 5b, 02-106 Warszawa, PL\\
{\tt msered@ippt.gov.pl}\\
{\tt www.ippt.gov.pl/$\sim$msered}\\ \mbox{ } \\
Andrzej Hanyga\\
ul. Bitwy Warszawskiej 14 m. 52, 02-366 Warszawa, PL\\
{\tt ajhbergen@yahoo.com}\\
{\tt www.uib.no/hjemmesider/andrzej}}

\title{Positivity of viscoelastic Green's functions} 

\begin{document}

\maketitle

\begin{abstract}
In 1,2 or 3 dimensions a scalar wave excited by a 
non-negative source in a viscoelastic medium with a non-negative relaxation spectrum or 
a Newtonian response or both combined inherits the sign of the source. The
key assumption is a constitutive relation which involves the sum of 
a Newtonian viscosity term and a memory term with a completely monotone relaxation kernel.
In higher-dimensional spaces this result holds for sufficiently regular sources. 
Two positivity results for vector-valued wave fields including isotropic viscoelasticity are 
also obtained. Positivity is also shown to hold under weakened hypotheses.
\end{abstract}


\noindent\textbf{Notation.}\\

\noindent $[a,b\,[\; := \{ x \in \mathbb{R} \mid a \leq x < b \}$;\\
$\mathbb{R}_+ =\; ]0,\infty[\;$;\\ 
$\mat{I}$: unit matrix; \\ \vspace{0.2cm}
$\langle \vec{k}, x\rangle := \sum_{n=1}^d k_l\, x^l$\,;\\ \vspace{0.2cm}
$\theta(y) := \left\{ \begin{array}{ll} 1, & y > 0 \\
0,  & y < 0 \end{array} \right\} $;\\ \vspace{0.2cm}
$\tilde{f}(p) := \int_0^\infty \e^{- p y} \, f(y) \, \dd y$; \\
$\mathfrak{L}$ -- the cone of LICM functions;\\
$\mathfrak{B}$ -- the cone of Bernstein functions;\\
$\mathfrak{F}$ -- the cone of complete Bernstein functions.

\section{Introduction}

Acoustic signals in real elastic and viscoelastic media reproduce the shape of the 
source pulse.
In particular, a non-negative source pulse gives rise to a non-negative signal.  
This property is important in acoustic and seismological applications.
We shall show that a large class of mathematical models of viscoelastic media,
including Newtonian viscosity and hereditary viscosity with positive relaxation spectrum,
have this property. A viscoelastic model obtained by connecting 
elastic, Newtonian and Debye elements in parallel also has this property.
Positivity of 
viscoelastic signals can be considered as a confirmation of the presence of these components.
Unfortunately, it follows from a remark in Sec.~\ref{sec:concluding} that 
a more general class of equations have non-negative Green's functions and 
some of them might be compatible with viscoelasticity.

Positivity of viscoelastic pulses was studied in a paper by Duff\cite{Duff69}. 
Duff assumed various special models with rational complex moduli. Duff's 
models are however loosely related to viscoelasticity and his assumptions are 
either excessively restrictive or inconsistent with viscoelasticity.

In Secs~\ref{sec:scalar1-3D}--\ref{sec:scalar-multidim}
a general scalar viscoelastic medium with the constitutive equation 
$\sigma = a \, \dot{e} + G(t)\ast \dot{e}$ with a completely monotone
relaxation modulus $G$ and a non-negative Newtonian viscosity coefficient $a$
is studied. We show that a scalar viscoelastic wave field 
propagating in a $d$-dimensional medium and 
excited by a non-negative pulse is also non-negative provided $d \leq 3$.
For higher dimensions and for non-zero initial data only wave fields excited by 
sufficiently regular sources are non-negative.

Positivity can be extended to matrix-valued fields, e.g. to Green's functions
of systems of PDEs. In Sec.~\ref{sec:vector} we consider a system of PDEs 
resembling the equations 
of motion of viscoelasticity with a CM relaxation kernel and prove that the Green's
function of this system of equations is positive semi-definite. This result 
does not apply to general viscoelastic Green's functions, which involve 
double gradients of positive semi-definite functions. The signals propagating in isotropic 
viscoelastic media can however be expressed in terms of scalar fields representing
the amplitudes or the potentials. It is thus again possible to prove that 
the scalar amplitudes excited by a non-negative force field are non-negative
(Sec.~\ref{sec:isotropic}).  

In Sec.~\ref{subsec:extended} positivity is proved under somewhat weaker 
assumptions, inconsistent with viscoelastic constitutive equations. In
Sec.~\ref{subsec:example} we examine an example of a scalar differential equation 
which has a positive Green's
function even when the assumptions of viscoelasticity with a positive relaxation
spectrum fail to be satisfied.

Our method takes advantage of certain classes of functions defined in terms of
integral transforms and, alternatively, in terms of analytic properties of their
complex analyitic continuations. Useful relations
between these classes can be established using the integral transforms. Analytic 
properties yield very useful nonlinear relations between these classes. These  
classes of functions appear in the viscoelastic constitutive 
equations~\citep{HanSerPRSA} but they also turn out very useful in studying positivity 
of Green's functions.

\section{Statement of the problem}
\label{sec:statement}

In a hereditary or Newtonian linear viscoelastic medium a scalar field 
excited by positive source is non-negative. This applies to displacements 
in pure shear or to scalar displacement potentials. The key assumption about 
the material properties of the medium is a positive relaxation spectrum.
The result holds for arbitrary spatial dimension. 

We consider the problem:
\begin{equation} \label{eq:problem}
\rho\, \D^2 \, u = a \, \nabla^2 \, \D\, u + G(t)\ast \nabla^2\,\D u + 
s(t,x)
 \quad t \geq 0, \quad x \in \mathbb{R}^d 
\end{equation}
with $s(t,x) = \theta(t)\,(c_1 + c_2 \, t) \,\delta(x)$ and the 
initial condition 
\begin{equation} \label{eq:ini}
u(0,x) = u_0\, \delta(x), \quad \D u(0,x) = \dot{u}_0 \,\delta(x),  
\end{equation}
(Problem~I) as well as $s(t,x) = c\,\delta(t) \, \delta(x)$ with a solution assumed to 
vanish for $t < 0$ (Problem~II).
It is assumed that $a \geq 0$ and $G$ is a completely monotone (CM) function.

The Laplace transform 
\begin{equation}
\tilde{u}(p,x) := \int_0^\infty \e^{-p t}\, u(t,x)\, \dd t, \quad \re p > 0,
\quad x \in \mathbb{R}^d
\end{equation}
satisfies the equation
\begin{equation} \label{eq:Laplacetransformed}
\rho\, p^2 \, \tilde{u}(p,x) = Q(p) \, \nabla^2\, \tilde{u}(p,x) 
+ g(p) \, \delta(x)
\end{equation}
where
\begin{equation} \label{eq:defQ}
Q(p) := a \, p + p \, \tilde{G}(p)
\end{equation}

In viscoelastic problems $a$ is the Newtonian viscosity. The instantaneous elastic response
is represented by $G_0 := \lim_{p\rightarrow\infty} [p \, \tilde{G}(p)]$. The equilibrium
response is represented by $G_\infty := \lim_{p\rightarrow 0} [p \, \tilde{G}(p)]$. 
Purely elastic response corresponds to $Q(p) = \const$. 

The function $g$ is defined by the equation
\begin{equation} 
g_I(p) = \frac{1}{p} + p\, u_0 + \dot{u}_0
\end{equation}
in Problem~I and
\begin{equation} 
g_{II}(p) = 1
\end{equation}
in Problem~II.

\section{Basic mathematical tools.}
\label{sec:tools}

The classes of functions appropriate for viscoelastic responses are 
reviewed in detail by Seredy\'{n}ska and Hanyga \citet{HanSerPRSA,HanMTAV}.

\begin{theorem}
If the function $\tilde{u}(\cdot,x)$ is completely monotone for every 
$x \in \mathbb{R}^d$, then $u(t,x) \geq 0$ for every  $t \geq 0$ and
$x \in \mathbb{R}^d$.
\end{theorem}
\begin{proof}
If $\tilde{u}(\cdot,x)$ is completely monotone, then, in view of Bernstein's theorem
(Theorem~\ref{thm:Bernstein}), for every $x \in \mathbb{R}^d$ it is the Laplace
transform of a positive Radon measure $m_x$:
\begin{equation} 
\tilde{u}(\cdot,x) = \int_{[0,\infty[} \e^{-p s}\, m_x(\dd s)
\end{equation}
The Radon measure $m_x$ is uniquely determined by $\tilde{u}(\cdot,x)$,
hence $m_x(\dd t) = u(t,x)\, \dd t$ is a positive Radon measure. Hence,
in view of continuity of $u(\cdot,x)$, we have the inequality 
$u(t,x) \geq 0$  for all $t \geq 0$ and $x \in \mathbb{R}^d$.  
\end{proof}

The problem of proving that $u(t,x)$ is non-negative is thus reduced to 
proving that $\tilde{u}(\cdot,x)$ is completely monotone. The crucial 
step here is the realization that $Q$ in \eqref{eq:defQ}
is a complete Bernstein function.
We shall therefore recall some facts about Bernstein and complete 
Bernstein functions and their relations to 
completely monotone functions.
\begin{definition}
A function $f$ on $\mathbb{R}_+$ is said to be \emph{completely monotone} (CM) if
it is infinitely differentiable and satisfies the infinite set of inequalities:
$$ (-1)^n \, \D^n \, f(y) \geq 0  \quad y > 0, \quad 
\text{for all non-negative integer $n$}$$
\end{definition}
It follows from the definition and the Leibniz formula that the
product of two CM functions is CM.
A CM function can have a singularity at 0. 
\begin{definition}
A function $f$ on $\mathbb{R}_+$ is said to be \emph{locally integrable 
completely monotone} (LICM) if it is CM and integrable over the segment $]0,1]$.
\end{definition}

Every CM function $f$ is the Laplace transform of a positive Radon measure:
\begin{theoremcited}(Bernstein's theorem, \citet{WidderLT}) \label{thm:Bernstein}
\begin{equation}  \label{eq:Bernstein}
f(t) = \int_{[0,\infty[}\; \e^{-r t} \, \mu(\dd r)
\end{equation} 
\end{theoremcited}
It is easy to show that $f$ is a LICM if the Radon measure
$\mu$ satisfies the inequality
\begin{equation} \label{eq:ineq}
\int_{[0,\infty[}\; \frac{\mu(\dd r)}{1 + r} < \infty
\end{equation}

\begin{theorem} \label{thm:superpositionandlimit}
\begin{enumerate}[(i)]
\item If $a, b > 0$ and the functions $f, g$ are CM then the function $a\, f + b\, g$ is 
CM. 
\item If $f_k$, $k=1,2,\ldots$ is a sequence of CM functions pointwise converging to $f$ then
$f$ is CM\citep{Jacob01I}. 
\end{enumerate}
\end{theorem}
The proof of the second statement is based on Bernstein's theorem
(Theorem~\ref{thm:Bernstein}).
Using Theorem~\ref{thm:superpositionandlimit} a very useful theorem can be proved:
\begin{theorem} \label{thm:integralsuperposition}
If 
\begin{enumerate}[1.]
\item $m$ is a positive measure on an interval $I$, 
\item the functions $f(\cdot,r)$ are defined and CM  for all $r$ in an interval $I$ except perhaps for a set of zero measure $m$,
 and
\item the integral $g(x) := \int_I f(x,r) \, m(\dd r)$ exists for $x > 0$, 
\end{enumerate}
then the function  
$g(x)$, $x > 0$, is CM.
\end{theorem}
The last theorem can be proved by approximating the integral by finite Riemann sums and applying Theorem~\ref{thm:superpositionandlimit}.
 
\begin{definition}
A function $f$ on $\mathbb{R}_+$ is said to be a \emph{Bernstein function} (BF) if
it is non-negative, differentiable and its derivative is a CM function.
\end{definition}
Since a BF is non-negative and non-decreasing, it has a finite limit at 0.
It can therefore be extended to a function on $\overline{\mathbb{R}_+}$.  

\begin{theoremcited} \label{thm:CMmult} \citep{Jacob01I,HanSerPRSA}\\
If $f, g$ are CM then the pointwise product $f\, g$ is CM.
\end{theoremcited}

Let $f$ be a Bernstein function. Since the derivative $\D f$ of $f$
is LICM, Bernstein's theorem can be applied. Upon integration the following
integral representation of a general Bernstein function $f$ is obtained:
\begin{equation}
f(y) = a + b\,y + \int_{]0,\infty[} \left[ 1 - \e^{-r \, y}\right] \,
\nu(\dd r) 
\end{equation}
where $a \geq 0$,  $b = \D f(0) \geq 0$, and $\nu(\dd r) := \mu(\dd r)/r$ is a 
positive Radon measure on $\mathbb{R}_+$ satisfying the inequality
\begin{equation} \label{eq:ineq2}
\int_{]0,\infty[}\; \frac{r \,\nu(\dd r)}{1 + r} < \infty
\end{equation}
The constants $a, b$ and the Radon measure $\nu$ are uniquely determined by the
function $f$.

\begin{theoremcited} \citep{BergForst,Jacob01I,HanMTAV}\\ \label{thm:compBFCM}
\begin{enumerate}[(i)]
\item If $f$ is a CM function, $g$ is a BF and $g(y) > 0$ for
$y > 0$ then the function $f\circ g$ is CM.
\item $f$ is a BF if and only if for every $\lambda > 0$ the function
$\exp(-\lambda \, f(x))$ is CM.
\end{enumerate}
\end{theoremcited} 

\begin{corollarycited} \citep{BergForst,Jacob01I,HanMTAV}\\ \label{corr:BF}
If $g$ is a non-zero BF then $1/g$ is a CM function.
\end{corollarycited} 

Note that the converse statement is not true: 
the function $f(y) := \exp(-y)$ is CM but $1/f$ is not a BF.

\begin{definition} \label{def:CBF}
A function $f$ is said to be a \emph{complete Bernstein function} (CBF) if 
there is a Bernstein function $g$ such that $f(y) = y^2\, \tilde{g}(y)$.
\end{definition}
\begin{theoremcited} \label{thm:analCBF}  \citep{Jacob01I}\\
A function $f$ is a CBF if and only if it satisfies the following two 
conditions:
\begin{enumerate}[(i)]
\item $f$ admits an analytic continuation $f(z)$ to the upper complex half-plane;
$f(z)$ is holomorphic and satisfies the inequality 
$\im f(z) \geq 0$ for $\im z > 0$;
\item $f(y) \geq 0$  for $y \in \mathbb{R}_+$.
\end{enumerate}
\end{theoremcited}

The derivative $\D g$ of the Bernstein function $g$ is a LICM function $h$. 
Hence we have the following theorem:
\begin{theoremcited} \citep{HanMTAV}
Every CBF $f$ can be expressed in the form 
\begin{equation} \label{eq:CBFLICM}
f(y) = y \, \tilde{h}(y) + a\, y
\end{equation}
where $h$ is LICM and $a = g(0) \geq 0$. Conversely, for every LICM function $h$
and $a \geq 0$ the function $f$ given by \eqref{eq:CBFLICM} is a CBF. 
\end{theoremcited}
\begin{proof} 
For the first part, let $g$ be the BF in Definition~\ref{def:CBF} and let $h := \D g$. 
Since $\int_0^1 h(x) \, \dd x = g(1) - g(0) < \infty$, the function $h$ is LICM.
For the second part, note that if $h$ is LICM, then 
$g(y) = a + \int_0^y h(s) \, \dd s$ is a BF and $f(y) = y^2 \, \tilde{g}(y)$.
\end{proof}

Since the Laplace transform of a LICM function $h$ has the form 
\begin{equation}
\tilde{h}(y) = \int_{[0,\infty[}\; \frac{\mu(\dd r)}{r + y}
\end{equation}
where $\mu$ is the Radon measure associated with $h$,
every CBF function $f$ has the following integral representation 
\begin{equation} \label{eq:CBFintegral}
f(y) = b + a \, y + y \int_{]0,\infty[}\; \frac{\mu(\dd r)}{r + y}
\end{equation}
with arbitrary $a, b  = \mu(\{0\}) \geq 0$ and an arbitrary positive 
Radon measure $\mu$ satisfying eq.~\eqref{eq:ineq}. The constants
$a, b$ and the Radon measure $\mu$ are uniquely determined by
the function $f$.

Noting that 
$y/(y + r) = r\, [ 1/r - 1/(y + r) ]$,
we can also express the CBF $f$ in the following form
$$f(y) = b + a\, y + \int_0^\infty \left[1 - \e^{-z \, y}\right]
h(z) \, \dd z$$
where $h(z) :=  \int_{]0,\infty[}\; \e^{-r z}\, m(\dd r) \geq 0$
and $m(\dd r) := r\,\mu(\dd r)$ satisfies the inequality
$$ \int_{[0,\infty[}\; \frac{m(\dd r)}{r (r + 1)} < \infty$$
Let $\nu(\dd z) := h(z) \, \dd z$. 
We have  
$$ \int_{[0,\infty[} \frac{z\,\nu(\dd z)}{1 + z} = 
\int_0^\infty \frac{z \, h(z)\, \dd z}{1 + z} = \\
\int_{[0,\infty[} m(\dd r) \,[1/r - e^r \, \Gamma(0,r)]$$
Using the asymptotic properties of the incomplete
Gamma function \citep{Abramowitz} it is possible to prove
that the right-hand side is finite, hence the Radon measure 
$\nu(\dd z) := h(z) \, \dd z$ satisfies inequality~\eqref{eq:ineq2}.
We have thus proved an important theorem:
\begin{theorem} \label{thm:BFCBF}
Every CBF is a BF.
\end{theorem}
However $1 - \exp(-y)$ is a BF but not a CBF.

The sets of LICM functions, BFs  and CBFs 
will be denoted by $\LICM$, $\BF$ and $\CBF$,
respectively.

The simplest example of a CBF is $$ \varphi_a(y) :=
y/(y+a) \equiv y^2 \int_0^\infty \e^{-s y} \left[ 1 - \e^{-s a}\right]
\, \dd s $$
$a \geq 0$. It follows from eq.~\eqref{eq:CBFintegral} that 
every CBF $f$ which satisfies the conditions 
$f(0) = 0$ and $\lim_{y\rightarrow \infty} f(y)/y = 0$ 
is an integral superposition of the functions $\varphi_a$. The CBF
$\varphi_a$ corresponds to a Debye element defined by the relaxation
function $G_a(t) = \exp(-a \,t)$.

We shall need the following properties of CBFs:
\begin{theoremcited} \label{thm:CBF} \citep{Jacob01I,HanSerPRSA}
\begin{enumerate}[(i)]
\item $f$ is a CBF if and only if $y/f(y)$ is a CBF;
\item if $f, g$ are CBFs, then $f \circ g$ is a CBF.
\end{enumerate}
\end{theoremcited}
The second statement follows easily from Theorem~\ref{thm:analCBF}. 
\begin{theorem} \label{thm:fracpow}
$y^\alpha$ is a CBF if $0 < \alpha < 1$. 
\end{theorem}
\begin{proof}
\begin{multline*}
y^{\alpha-1} = \frac{1}{\Gamma(1-\alpha)} 
\int_0^\infty \e^{-y s}\, s^{-\alpha} \, 
\dd s =\\ \frac{1}{\Gamma(1-\alpha) \, \Gamma(\alpha)} 
\int_0^\infty \dd z \int_0^\infty  \,\e^{-y s} \, \e^{-z s} \dd s\;
z^{\alpha-1} = \\ 
\frac{\sin(\alpha \, \upi)}{\upi}
\int_0^\infty \frac{z^{\alpha-1}}{y + z}\, \dd z 
\end{multline*}
and thus 
$$y^\alpha = \frac{y\, \sin(\alpha \, \upi)}{\upi}
\int_0^\infty \frac{z^{\alpha-1}}{y + z}\, \dd z$$
\end{proof}

Using Theorem~\ref{thm:CBF}(ii), the last theorem implies the following corollary:
\begin{corollary}
If $f$ is a CBF and $0 \leq \alpha \leq 1$, then $f(x)^\alpha$ is a CBF.
\end{corollary}

\begin{remark}
If $f$ is CM then $f^n$ is CM for positive integer $n$ (Theorem~\ref{thm:CMmult}).
However $f^\alpha$ need not be CM for non-integer $\alpha$.
\\
If $f \in \CBF$ then $f^\alpha$ need not be a CBF for $\alpha > 0$ (including
integer values) unless $0 \leq \alpha \leq 1$. 
\end{remark}

\section{Positivity of solutions in one- and three-dimensional space.}
\label{sec:scalar1-3D}

Applying the results of the previous section, we get the following result:
\begin{theorem}
If $a \geq 0$ and the relaxation modulus $G$ is CM then the function $Q$
defined by eq.~\eqref{eq:defQ} is a CBF.\\
The mapping $(a, G) \in \overline{\mathbb{R}_+}\times \mathfrak{F} 
\rightarrow Q \in \mathfrak{L}$ defined by eq.~\eqref{eq:defQ} is bijective.
\end{theorem}

A one-dimensional solution of eq.\eqref{eq:Laplacetransformed} is given by
$$\tilde{u}^{(1)}(p,x) = U^{(1)}(p,\vert x\vert):=  
A(p) \, \exp(-B(p) \, \vert x \vert)$$ with
$B(p) = \rho^{1/2}\,p/Q(p)^{1/2}$ and $A(p) = g(p)/[2 B(p)]$.
If $Q \in \mathfrak{F}$, then $Q(y)^{1/2}$ is a composition of two CBFs, namely 
$y^{1/2}$ (Theorem~\ref{thm:fracpow}) and $Q$, hence it is a CBF
by Theorem~\ref{thm:CBF}. The function $B(p)$ is a CBF by Theorem~\ref{thm:CBF}
and $1/B(p)$ is a CM function by Theorem~\ref{thm:BFCBF} and 
Corollary~\ref{corr:BF}. 

The amplitude of the solution of Problem~I is given by 
$A(p) = 1/[2\, p \,B(p)] + \dot{u}_0/[2 B(p)]$. The first term 
is a CM function because it is the product of two CM functions. The second 
term is also CM, hence $A(p)$ is CM. The amplitude of the solution of Problem~II 
$A(p) = 1/[2\, B(p)]$ is also CM.

For every fixed $x$ the function $\exp(-B(p)\, \vert x \vert)$ is the composition
of a CBF and the function $B$, which is a CBF and therefore a BF. 
By Theorem~\ref{thm:compBFCM} the function $\exp(-B(\cdot)\, \vert x \vert)$ is
CM. This proves that for $d = 1$ the solutions of Problem~II and Problem~I with
$u_0 = 0$ are non-negative.

In a three-dimensional space the solution $\tilde{u}^{(3)}(p,x) = U^{(3)}(p,\vert x \vert)$ of 
\eqref{eq:Laplacetransformed} is given by the equation 
$$U^{(3)}(p,r) = -\frac{1}{2 \upi r} \frac{\partial U^{(1)}(p,r)}{\partial r}, \quad r > 0$$
\citep{HanMTAV} so that 
\begin{equation}
\tilde{u}^{(3)}(p,x) = \frac{1}{4 \upi r} A(p) \, B(p) \, 
\exp(-B(p) \, \vert x \vert)
\end{equation}
But $A(p) \, B(p) = g(p)/2$. If $u_0 = 0$ then $g$ is CM. Hence
$\tilde{u}^{(3)}(\cdot,x)$ is the product of two CM functions and thus CM.

\section{Positivity of scalar Green's functions in arbitrary dimension.}
\label{sec:scalar-multidim}

In an arbitrary dimension $d$ 
\begin{equation}
\tilde{u}^{(d)}(p,x) = \frac{g(p)}{(2 \upi)^d\, Q(p)} 
\int \e^{\ii \langle \vec{k},x\rangle}
\frac{1}{\rho\, p^2/Q(p) + \vert \vec{k} \vert^2} \dd_d k
\end{equation}
The above formula can be expressed in terms of MacDonald functions
by using eq.~(3) in Sec.~3.2.8 of \citet{Gelfand}:
\begin{equation}
\tilde{u}^{(d)}(p,x) = 
\frac{\rho^{d/4-1/2} \, g(p) \,p^{d/2-1}}{(2 \upi)^{d/2}\, Q(p)^{d/4+1/2}}
r^{-(d/2-1)}\,
K_{d/2-1}\left(B(p)\, r\right)
\end{equation}
where $B(p)$ is defined in the preceding section. 

The MacDonald function is given by the integral representation
\begin{equation}
K_\mu(z) = \int_0^\infty \exp(-z\, \cosh(s)) \, \cosh(\mu s) \, \dd s
\end{equation}
Since $\cosh(y)$ is a positive increasing function, it follows immediately 
that $K_\mu(z)$ is a CM function.

We shall need a stronger theorem on complete monotonicity of MacDonald functions. 
\begin{theoremcited} \citep{MillerSamko01}.\\ \label{thm:McD}
The function $L_\mu(z) := z^{1/2}\, K_\mu(z)$ is CM for $\mu \geq 1/2$.
\footnote{The theorem is valid for $\mu \geq 0$, see \citep{MillerSamko01},
but we do not need this fact.}
\end{theoremcited}
The proof of this theorem requires a lemma.
\begin{lemma} \label{lem}
If $\alpha \geq 0$ then the function $\left( 1 + 1/x \right)^\alpha$ is CM.
\end{lemma}
\begin{proof}
We begin with $0 \leq \alpha < 1$. Setting $t = 1/(x y)$ we have that
\begin{multline*}
\frac{\alpha}{x^\alpha} \int_1^\infty \frac{\dd y}{y^{1+\alpha} \,(xy + 1)^{1-\alpha}} = 
\alpha \int_0^{1/x} \frac{t^{\alpha-1}}{(1/t + 1)^{1-\alpha}} \, \dd t = \\
\alpha \int_1^{1+1/x} u^{\alpha-1}\, \dd u = \left(1 + \frac{1}{x}\right)^\alpha
\end{multline*}
Since for each fixed value of $y > 0$ the function $(x y + 1)^{\alpha-1}$ is CM,
the function $\left(1 + 1/x\right)^\alpha$ ($x > 0)$ is also CM.

The function $1 + 1/x$ is CM, hence for every positive integer $n$ 
the function $(1 + 1/x)^n$ is CM. We can now decompose any positive non-integer 
$\alpha$ into the sum $\alpha = n + \beta$, where $n$ is a positive integer 
and $0 < \beta < 1$. Consequently
$$(1 + 1/x)^\alpha \equiv (1 + 1/x)^n \, (1 + 1/x)^\beta$$
is CM because it is a product of two CM functions.
\end{proof}

\begin{prooftheorem} 
For $\mu > -1/2$ the function $L_\mu(z)$ has the following integral representation:
\begin{equation}
L_\mu(z) = \sqrt{\frac{\pi}{2}} \frac{1}{\Gamma(1/2-\mu)}
\e^{-z}\times\\ \int_0^\infty \e^{-s} \, s^{\mu-1/2} \,
\left(1 + \frac{s}{2 z}   \right)^{\mu-1/2} \, \dd s, \qquad z > 0
\end{equation}
\citep{GradshteinRhyzhik}, 8:432:8. By Lemma~\ref{lem}
the integrand of the integral on the right-hand side is CM if $\mu \geq 1/2$. 
Hence the integral is the limit of 
sums of CM functions, therefore itself a CM function. Consequently, the function 
$L_\mu(z)$ is the product of two CM functions. By Theorem~\ref{thm:CMmult} $L_\mu$ is CM.
\end{prooftheorem}

We now note that 
$\tilde{u}^{(d)}(p,x) = C\, p^{(d-3)/2} \, g(p)\, F(p)$, 
where $C = \rho^{d/4-1/2}/(2 \upi)^{d/2}$
and $F(p) =  Q(p)^{-(d+1)/4}\, L_{d/2-1}(B(p)\,r)$.  We shall prove that $F(p)$ is 
the product of two 
CM functions.
\begin{lemma} \label{lem:Q}
If $Q$ is a CBF and $\alpha > 0$ then $Q(p)^{-\alpha}$ is CM.
\end{lemma}
\begin{proof}
Let $n$ be the integer part of $\alpha$, $\alpha = n + \beta$,
$0 \geq \beta < 1$. $Q(p)^{-1}$ is CM (by Theorem~\ref{thm:BFCBF} and 
Corollary~\ref{corr:BF}) and therefore also $Q(p)^{-n}$ is CM. By 
Theorem~\ref{thm:CBF} the function $Q(p)^\beta$ is a CBF,
hence $1/Q(p)^\beta$ is CM. Consequently $Q(p)^{-\alpha}$ is CM.
\end{proof}

\begin{lemma} \label{lem:basicK}
If $Q \in \CBF$ and $\mu \geq 1/2$ then $L_{d/2-1}(B(p)\,r)$ is CM.
\end{lemma}
\begin{proof}
We have already proved that $B$ is a BF. 
It follows from Theorem~\ref{thm:McD} and Theorem~\ref{thm:compBFCM} 
that the function $L_\mu(B(p) \, r)$ is CM for $\mu \geq 1/2$.
\end{proof}

Lemma~\ref{lem:Q}  implies that for every $d > 0$ the factor $Q(p)^{-(d+1)/4}$ is CM
while Lemma~\ref{lem:basicK} implies that $L_{d/2-1}(B(p)\,r)$ is a CM function of $p$ for every $r > 0$ 
for every $d \geq 1$.
Hence the function $F(p)$ is CM. The factor $p^{(d-3)/2}$ is CM if $d \leq 3$.
Consequently,
for $d \leq 3$ the solution $u(t,x)$ of Problem~II is non-negative. The 
solution of the same problem with an arbitrary source of the form 
$s(t)\, \delta(x)$ and $s(t) \geq 0$ can be obtained by a convolution of 
two non-negative functions and therefore is also non-negative.

For $d \leq 5$  Problem~I with $u_0 = \dot{u}_0 = 0$ has a non-negative 
solution if $c_1 > 0$. For $d \leq 7$ Problem~I has a non-negative solution if
$c_1 = 0$ and $c_2 > 0$.

We summarize these results in a theorem.

\begin{theorem}
In a viscoelastic medium of dimension $d \leq 3$ with a constitutive relation
$$\sigma = a\, \dot{e} + G(t)\ast\dot{e}, \qquad a \geq 0; \quad G \in 
\mathfrak{F}$$
Problem~II as well as Problem~I with the initial condition $u_0 = 0$ have 
non-negative solutions. 

For zero initial data Problem~I has a non-negative solution
if $d \leq 5$ and $c_1 > 0$, or if $d \leq 7$, $c_1 = 0$ and $c_2 > 0$.
\end{theorem}

Anisotropic effects can be introduced by replacing 
the operator $\nabla^2$ by $g^{kl} \,\partial_k \, \partial_l$.
If $h_{kl}\, g^{lm} = \delta^k_m$ then
\begin{multline}
\tilde{u}^{(d)}(p,x) = \sqrt{\det{g}} \;
\frac{\rho^{d/4-1/2} \, g(p) \,p^{d/2-1}}{(2 \upi)^{d/2}\, Q(p)^{d/4+1/2}}
r^{-(d/2-1)}\times\\
K_{d/2-1}\left(\rho^{1/2}\, p   r/Q(p)^{1/2}\right)
\end{multline}
If $Q$ is a CBF then $u(t,x) \geq 0$,  
with $r := \left[h_{kl} \, x^k \, x^l\right]^{1/2}$,
cf \citet{Gelfand}.

\section{Positivity properties of vector-valued fields.}
\label{sec:vector}

It is interesting to examine the implications of CM relaxation kernels on 
positivity properties of vector fields.
We shall prove that in a simple model complete monotonicity of a relaxation kernel 
implies that the Green's function is positive semi-definite.

Unfortunately the tools developed in Sec.~\ref{sec:tools} 
fail for matrix-valued CM and complete Bernstein functions.
In particular, the product of two non-commuting matrix-valued functions need not be a CM function and
the function $f\circ \mat{G}$, where $f$ is CM and $\mat{G}$ is a matrix-valued BF, need not be CM.

\begin{definition}
A matrix-valued function  $\mat{F}: \mathbb{R}_+ \rightarrow \mathbb{R}^{n\times n}$ is said to be a 
CM function if it is infinitely differentiable and the matrices $(-1)^n\, \D^n \, \mat{F}(y)$ are 
positive semi-definite for all $y > 0$.
\end{definition}
\begin{definition}
A matrix-valued Radon measure $\mat{M}$ is said to be positive if 
$$\langle \vec{v}, \int_{[0,\infty[}\; f(y) \, \mat{M}(\dd y)\, \vec{v} \rangle$$
for every vector $\vec{v} \in \mathbb{R}^n$ and every non-negative function $f$ 
on $\overline{\mathbb{R}_+}$ with compact support. 
\end{definition}
It is convenient to eliminate matrix-valued Radon measures by applying the following lemma 
\citep{HanDuality}:
\begin{lemma} 
Every matrix-valued Radon measure $\mat{M}$ has the form 
$\mat{M}(\dd x)=\mat{K}(x)\, m(\dd x)$, 
where $m$ is a positive Radon measure, 
while $\mat{K}$ is a matrix-valued function defined, bounded and positive semi-definite on 
$\mathbb{R}_+$ except on a subset $E$
such that $m(E) = 0$. 
\end{lemma}

\begin{theoremcited} \citep{GripenbergLondenStaffans}\\
A matrix-valued function $\mat{F}: \mathbb{R}_+ \rightarrow \mathbb{R}^{n\times n}$ is CM 
if and only if it is the Laplace transform of a positive matrix-valued Radon measure.
\end{theoremcited}
The following corollary will be applied to Green's functions:
\begin{corollary}
If $\tilde{\mat{R}}(p) := \int_0^\infty \e^{-p t} \, \mat{R}(t) \, \dd t$ is a matrix-valued CM function then
$\mat{R}(t)$ is positive semi-definite for $t > 0$.
\end{corollary}

\begin{definition}
A matrix-valued function  $\mat{G}: \mathbb{R}_+ \rightarrow \mathbb{R}^{n\times n}$ 
is said to be a Bernstein function (BF)
if $\mat{G}(y)$  is differentiable and positive semi-definite for all $y > 0$ and 
its derivative $\D \mat{G}$ is CM.
\end{definition}

\begin{definition}
A matrix-valued function  $\mat{H}: \mathbb{R}_+ \rightarrow \mathbb{R}^{n\times n}$ is 
said to be a complete Bernstein function (CBF)
if $\mat{H}(y) = y^2\, \tilde{\mat{G}}(y)$, where $\mat{G}$ is an $n \times n$ matrix-valued BF.
\end{definition}

The integral representation \eqref{eq:CBFintegral} of a CBF remains valid except that 
the Radon measure has to be replaced by a positive matrix-valued
Radon measure $\mat{N}(\dd r) = \mat{K}(r) \,\nu(\dd r)$:
\begin{equation} \label{eq:CBFintegralM}
\mat{H}(y) = \mat{B} +  y\, \mat{A} + y \int_{]0,\infty[}\; \frac{\mat{K}(r)\,\nu(\dd r)}{r + y}
\end{equation}
where the Radon measure $\nu$ satisfies the inequality 
\begin{equation} 
\int_{[0,\infty[}\; \frac{\nu(\dd r)}{1 + r} < \infty
\end{equation}
the matrix-valued function $\mat{K}(r)$ is positive semi-definite and bounded 
$\nu$-almost everywhere on $\mathbb{R}_+$
while $\mat{A}$, $\mat{B}$ are two positive semi-definite matrices. 
Every matrix-valued CBF $\mat{H}$ can be expressed in the form 
\begin{equation} \label{eq:CBFLICM-M}
\mat{H}(y) = y \, \tilde{\mat{F}}(y) +  y\, \mat{A}
\end{equation}
where $\mat{F}$ is a matrix-valued LICM function.

We now consider the following problem for the Green's function $\mathcal{G}$
\begin{equation}
\rho\, \D^2 \,\mathcal{G} = \mat{A}\,\D\, \mathcal{G} +  \mat{G}\ast \nabla^2\, 
\D\, \mathcal{G} + \delta(t) \,\delta(x)\,\mat{I}, 
\quad t \geq 0, \quad x \in \mathbb{R}^d
\end{equation}
where $\mat{A}$ is a positive semi-definite $n \times n$ matrix and $\mat{G}$ is 
an $n \times n$ matrix-valued relaxation modulus.

If the relaxation modulus $\mat{G}$ is a CM matrix-valued function then the function 
$\mat{Q}(p) := p \, \tilde{\mat{G}}(p)$ 
is a matrix-valued CBF. The function $\mat{Q}$ is real and positive semi-definite, 
hence it is symmetric and has $n$ 
eigenvalues $q_i(p)$ and $n$ eigenvectors $\vec{e}_i$, $i = 1, \ldots, n$. 
We shall now assume that the eigenvectors are constant:
$$\mat{Q}(p) = \sum_{i=1}^n q_i(p) \, \vec{e}_i \otimes  \vec{e}_i$$

It is easy to see that the functions $q_i$, $i = 1,\ldots, n$, are CBFs.

The Laplace transform $\tilde{\mathcal{G}}(p,x)$ of the Green's function is given by the
formula
\begin{multline*}
\tilde{\mathcal{G}}(p,x) = \frac{1}{(2 \upi)^d} \int \e^{\ii \langle \vec{k}, x \rangle} \left[ p^2\, \mat{I} + 
\vert \vec{k} \vert^2\, \mat{Q}(p)\right]^{-1} \, \dd_d k \equiv \\ \sum_{i=1}^n \frac{1}{(2 \upi)^d} \int \e^{\ii \langle \vec{k}, x \rangle} 
 \left[ p^2 + \vert \vec{k} \vert^2 \, q_i(p)\right]^{-1} \,\vec{e}_i \otimes  \vec{e}_i \dd_d k \equiv \\ 
\sum_{i=1}^d  g_i(p,x) \, \vec{e}_i \otimes  \vec{e}_i
\end{multline*}
where 
$$g_i(p,x) := \frac{\rho^{d/4-1/2} \,p^{d/2-1}}{(2 \upi)^{d/2}\, q_i(p)^{d/4+1/2}}\, 
r^{-(d/2-1)}\, K_{d/2-1}\left(B_i(p)\, r\right)$$
and
$B_i(p) = \rho^{1/2}\,p/q_i(p)^{1/2}$, $i = 1, \ldots, n$. Assume for definiteness that  $d \leq 3$. The argument of Sec.~\ref{sec:scalar-multidim} 
now leads to the conclusion that 
the functions $g_i$, $i = 1, \ldots, n$, are CM. Consequently the Laplace transform 
$\tilde{\mathcal{G}}(\cdot,x)$ of the Green's function is a matrix-valued CM function for $x \in \mathbb{R}^d$
and the Green's function $\mathcal{G}(t,x)$ is positive semi-definite for $t \geq 0$, $x \in \mathbb{R}^d$. 
In particular, we have the following theorem:
\begin{theorem}
Let $\rho \in \mathbb{R}_+$, $d \leq 3$, $\vec{s}(t,x) = \delta(t)\, \delta(x)\, \vec{w}$, where $\vec{w} \in \mathbb{R}^n$. 

If $\mat{G}(s) = \sum_{i=1}^n G_i(s) \, \vec{e}_i \otimes  \vec{e}_i$ and $\mat{A} = \sum_{i=1}^n a_i \, \vec{e}_i \otimes  \vec{e}_i$
with CM functions $G_k$ and real numbers $a_i \geq 0$, $i = 1, \ldots, n$, then 
the solution $\vec{u}$ of the problem
$$\rho\, \D^2\, \vec{u} = \mat{A} \, \D\, \vec{u} + \mat{G}\ast \nabla^2 \, \D\,\vec{u} + \vec{s}(t,x), \quad t\geq 0, \quad x \in \mathbb{R}^d$$ 
satisfies the inequality
\begin{equation}
\langle \vec{u}(t,x), \vec{w} \rangle \geq 0, \qquad t \geq 0, \quad x  \in \mathbb{R}^d.
\end{equation}
\end{theorem}

\section{Positivity in isotropic viscoelasticity.}
\label{sec:isotropic}

Consider now the Green's function $\mathcal{G}$ of a 3D isotropic viscoelastic medium. 
The function $\mathcal{G}$ 
is the solution of the initial-value 
problem:
\begin{multline}
\rho \, \D^2 \, \mathcal{G}_{kr}(t,x) = G(t)_{klmn}\ast \D\,\mathcal{G}_{mr,nl} + \delta(t)\, \delta(x) \, \delta_{kr}, \\ t > 0, 
\quad x \in \mathbb{R}^3, \quad k,r = 1,2,3
\end{multline}
with zero initial conditions, and
\begin{equation}
G_{klmn}(t) = \lambda(t) \, \delta_{kl}\, \delta_{mn} + \mu(t) \, \delta_{km}\,
\delta_{ln} + \mu(t) \, \delta_{kn}\, \delta_{lm}
\end{equation}
where the kernels $\lambda(t), \mu(t)$ are CM and $\rho \in \mathbb{R}_+$. The function $\tens{G}$ with the components $G_{klmn}$ 
takes values in the linear space $\Sigma$ of
symmetric operators on the space $S$ of symmetric $3 \times 3$ matrices. It is easy to see that under our hypotheses this function is CM:
$$ (-1)^n \, \langle \mat{e}_1, \tens{G}(t)\, \mat{e}_2 \rangle \geq 0 \qquad \text{for all $n = 0,1, 2 \ldots$} $$ 
for every $\mat{e}_1, \mat{e}_2 \in S$, where $\langle \mat{v}, \mat{w} \rangle := v_{kl}\, w_{kl}$ is the inner product on $S$. 

The Laplace transform $\tilde{\mathcal{G}}$ of $\mathcal{G}$ is given by the formula
\begin{multline*}
\mathcal{G}(p,x) = \frac{1}{p\, \rho}  \Big\{  
\nabla \otimes \nabla \Delta^{-1}\, F_\mathrm{L}(p,\vert x \vert)
+ \\ \left[ \mat{I} -  \nabla \otimes \nabla \Delta^{-1}\right]\, F_\mathrm{T}(p,\vert x \vert) 
\Big\}
\end{multline*}
where $\Delta := \nabla^2$, 
\begin{gather}
F_{\mathrm{L}}(p,r) := \frac{s_{\mathrm{L}}(p)^2 \,}{4 \upi r} \e^{-p^{1/2}\, s_{\mathrm{L}}(p)\, r} \label{eq:qq}\\
F_{\mathrm{T}}(p,r) := \frac{s_{\mathrm{T}}(p)^2 \,}{4 \upi r} \e^{-p^{1/2}\, s_{\mathrm{T}}(p)\, r} \label{eq:qr}
\end{gather}
and
\begin{gather}
s_{\mathrm{L}}(p)^2 := \frac{\rho}{\lambda(p) + 2 \mu(p)}\\
s_{\mathrm{T}}(p)^2 := \frac{\rho}{\mu(p)}
\end{gather}
Since $q_{\mathrm{L}}(p) = p/s_{\mathrm{L}}(p)$ and $q_{\mathrm{T}}(p) = p/s_{\mathrm{T}}(p)$ 
are CBFs, the functions $p/q_{\mathrm{L}}(p)^{1/2} =
p^{1/2}\, s_{\mathrm{L}}(p)$ and $p/q_{\mathrm{T}}(p)^{1/2} =
p^{1/2}\, s_{\mathrm{T}}(p)$ are BFs. Hence the exponentials in eqs~(\ref{eq:qq}--\ref{eq:qr}) 
are CM functions of $p$.
 Moreover the functions $p \,s_{\mathrm{L}}(p)^{-2}$, $p \,s_{\mathrm{T}}(p)^{-2}$ are CBFs  
and thus $s_{\mathrm{L}}(p)^2/p$ and $s_{\mathrm{T}}(p)^2/p$ are CM. It follows that
the functions $F_{\mathrm{L}}(p,r)$ and $F_{\mathrm{T}}(p,r)$ are CM and therefore they are 
Laplace transforms of non-negative
functions $F_{\mathrm{L}}(t,r)$ and $F_{\mathrm{T}}(t,r)$. Their indefinite integrals 
$f_{\mathrm{L}}(t,r) := \int_0^t F_{\mathrm{L}}(s,r)\,
\dd s$ and $f_{\mathrm{T}}(t,r) := \int_0^t F_{\mathrm{L}}(s,r)\,
\dd s$ are also non-negative. 
The functions $h_{\mathrm{L}}(t,r) := \Delta^{-1}\, f_{\mathrm{L}}(t,r)\,
\dd s$, 
$h_{\mathrm{T}}(t,r) = \Delta^{-1}\, f_{\mathrm{T}}(t,r)$
 involve a convolution with a non-negative kernel and therefore 
are non-negative. The Green's function can be expressed in terms of these functions:
\begin{equation} \label{eq:GreenIsoVE}
\mathcal{G}(t,x) = \frac{1}{\rho} \left\{ \nabla \otimes \nabla h_{\mathrm{L}}(t,\vert x \vert) + 
\left[ \Delta\, \mat{I} - \nabla \otimes \nabla \right]\, h_\mathrm{T}(t,\vert x \vert) \right\}
\end{equation}

We shall use the notation $\vec{v} \geq 0$ if $v_k \geq 0$ for $k =1,2,3.$

\begin{theorem}
Let $\vec{u} = \nabla \phi + \nabla \times \vec{\psi}$ be the solution of the initial-value problem 
\begin{equation}
\rho\, \D^2 \,\vec{u} = \mat{G}\ast \nabla^2\, \D\, \vec{u} + \vec{s}(t,x), 
\quad t \geq 0, \quad x \in \mathbb{R}^d
\end{equation}
with $\vec{u}(0,x) = 0 = \D \vec{u}(0,x)$ and
$\vec{s}(t,x) = \nabla f(t,x) + \nabla \times \vec{g}(t,x)$.

Then $\nabla f(t,x) \geq 0$ for all $t \geq 0$, $x \in \mathbb{R}^3$ implies 
that $\nabla \phi(t,x) \geq 0$ for all $t \geq 0$, $x \in \mathbb{R}^3$.

Similarly, $\nabla \times \vec{g}(t,x) \geq 0$ for all $t \geq 0$, $x \in \mathbb{R}^3$ implies 
that $\nabla \times \vec{\psi}(t,x) \geq 0$ for all $t \geq 0$, $x \in \mathbb{R}^3$.
\end{theorem}
\begin{proof}
Substitute $\vec{u} = \nabla \phi + \nabla \times \vec{\psi}$, 
$\vec{s} = \nabla f + \nabla \times \vec{g}$ in the formula
$$\vec{u}(t,x) = \int_0^t \int \mathcal{G}(t-s,x-y)\, \vec{s}(s,y) \dd_3 y \,
\dd s$$
where $\mathcal{G}$ is given by \eqref{eq:GreenIsoVE}. Noting that $\Delta^{-1}$
is a convolution operator commuting with $\nabla$ and 
$\nabla \,\Delta^{-1}\, \Div \vec{s} = \nabla f$ we have
$$\nabla \phi(t,x) = \frac{1}{\rho} \, \int f_{\mathrm{L}}(t-s,\vert x-y \vert)\,
(\nabla f)(s,y)\, \dd_3 y$$
We now note that
$\left[\mat{I} - \Delta^{-1}\, \nabla \otimes \nabla\right] \, \vec{s} =
\vec{s} - \nabla f = \nabla \times \vec{g}$. Hence  
$$\nabla \times \vec{\psi}(t,x) = \frac{1}{\rho} \, 
\int f_{\mathrm{T}}(t-s,\vert x-y \vert)\, \nabla \times \vec{g}(s,y) \, \dd_3 y$$
The functions $f_{\mathrm{L}}$ and $f_{\mathrm{T}}$ are non-negative, hence the thesis follows.
\end{proof}

\section{Generalized Cole-Cole relaxation modulus.}

A particular example of a CBF is the rational function
$F(p) = R_N(p)/S_M(p)$, where $R_N$ and $S_M$ are two polynomials with 
simple negative roots $\lambda_k$, $k = 1, \ldots, N$, $\mu_l$,
$l = 1,\ldots, M$, $M = N$ or $N+1$ satisfying the intertwining conditions:
$$0 \leq \lambda_1 < \mu_1 < \ldots \mu_N \,[\, < \lambda_{N+1}\,]$$
(the last inequality is applicable only if $M = N + 1$). The function 
$F$ is thus a good example for the function $Q$.

 A more general CBF is obtained by substituting in $F$ the
CBF $p^\alpha$, with $0 < \alpha < 1$:
$$ F_\alpha(p) = R_N\left(p^\alpha\right)/S_M\left(p^\alpha\right)$$
(Theorem~\ref{thm:CBF}). The choice of 
$Q = F_\alpha$ corresponds to a generalized Cole-Cole model
of relaxation. For $N = M = 1$ the original Cole-Cole model  
\citep{ColeCole,BagleyTorvik1,HanMTAV,HanColeJCA} is recovered.

Another useful CBF is the function $F_{\tau,\alpha}(p) := 
F\left(\left(1 + (\tau p)\right)^\alpha\right)$,
for some $\alpha \in [0,1]$ and $\tau > 0$. The function $\left(1 + (\tau p)\right)^\alpha$
is a CBF, hence $F_{\tau,\alpha} \in \CBF$ too.
This function has a low-frequency behavior appropriate for 
poroacoustic and poroelastic media. A simple poro-acoustic model of this type
is defined by the equation\citep{BiotI,Norris1:BIOT,Allard:BIOT,AuriaAl:BIOT,Johnson,PrideGangiDaleMartin,HanGreen,FellahDepollier} 
\begin{equation}\label{eq:poro}
\rho\, p^2\, K(p)\, \tilde{u}(p,x) = \Delta \, \tilde{u}(p,x) + \tilde{s}(p,x)
\end{equation}
where $K(p) := F_{\tau,\alpha}(p)^{-1}$ is a CM function, known as the viscodynamic
operator\cite{Norris1:BIOT}.
Bernstein's theorem implies
that the inverse Laplace transform $H(t)$ of $K(p)$ is non-negative. In terms of the kernel 
$H(t)$ eq.~\eqref{eq:poro} is equivalent to the time-domain equation
\begin{equation} \label{eq:poro-time}
\rho\, \D^2\, H(t)\ast u(t,x) = \Delta \, u(t,x) + s(t,x)
\end{equation}
 
The results of Sec.~\ref{sec:scalar-multidim} imply that for $Q = F, F_\alpha, F_{\alpha,\tau}$, 
and $0 \leq \alpha \leq 1$, $\tau > 0$, $d \leq 3$
the Green's functions of \eqref{eq:problem} and \eqref{eq:poro-time} are non-negative for 
$d \leq 3$.

\section{Beyond viscoelasticity.}

\subsection{Positivity for $Q^{1/2} \in \CBF$.} 
\label{subsec:extended}

The proofs of positivity in the preceding sections were motivated by the constitutive 
equations of viscoelasticity which led to the assumption $Q \in \CBF$ and its variants
underlying all the proofs. We now approach the same problem from the other end and ask
whether positivity would hold under weaker assumptions than $Q \in \CBF$.

Positivity can be shown to hold under somewhat weaker assumptions. The Green's function for 
eq.~\eqref{eq:problem} can be expressed in the form
\begin{multline} \label{eq:add}
\tilde{U}^{(d)}(p,r) = \frac{1}{(2 \upi)^{d/2} \, \rho} p^{(d-3)/2}\, \left[\frac{B(p)}{p}  
\right]^{(d+1)/2}\, r^{-(d-1)/2} \, L_{d/2-1}(B(p)\, r)
\end{multline}
For $d = 1,3$ the function $L_{d/2-1}(z) = (\upi/2)^{1/2} \, \exp(-z)$ and 
the last factor in \eqref{eq:add} is CM for every $r$ if and only if $B$ is a BF
(Theorem~\ref{thm:compBFCM}). On the other hand,
if $B$ is a BF then, by Lemma~\ref{lem:app}, the function $B(p)/p$ is CM. If $d$ is an 
odd integer then the third factor on the right-hand side of
\eqref{eq:add} is CM. The second factor is CM if $d \leq 3$. Therefore the Green's function is
non-negative if $B$ is a BF and $d = 1,3$. 

Positivity of $u^{(2)}$ can be proved under somewhat stronger assumption 
$B \in \CBF$. Let $0 \leq \gamma < 1$. Since $B, p^\gamma \in \CBF$,
from Theorem~\ref{thm:logconvCBF} we conclude that 
$$\left[\frac{p}{B(p)}\right]^{1/2}\, p^{\gamma/2}  \in \CBF$$
Hence 
$$f_\gamma(p) := \left[\frac{B(p)}{p}\right]^{3/2} \, p^{-\gamma/2} = \frac{B(p)}{p} 
\left[\frac{B(p)}{p}\right]^{1/2} \, p^{-\gamma/2}$$
is the product of two CM functions. Hence it is CM. 
The pointwise limit $\lim_{\gamma \rightarrow 1} f_\gamma(p)$ is CM.  
Since $L_0(B(p) r)$ is CM, eq.~\eqref{eq:add} implies that 
$U^{(2)}(\cdot,r)$ is CM for all $r > 0$. Hence $u^{(2)}(t,x) \geq 0$.

If $B \in \CBF$ then the function $Q(p)^{1/2} = p/B(p)$ is also a CBF but $Q$ need not be a CBF. 
Since $Q \in \CBF$ entails $Q^{1/2} \in \CBF$ but not conversely, 
the assumption that $B \in \CBF$ or even $B \in \BF$ is weaker than the assumption $Q \in \CBF$.
For $d = 1$ and 3 this constitutes a generalization of the results obtained in
Sec.~\ref{sec:scalar-multidim}.

\subsection{Example.}
\label{subsec:example}

We shall consider viscoelastic media with power law attenuation,
defined by the formula \citep{HanPowerlaw,HanMTAV}
\begin{equation}
Q(p) = \frac{c^2}{[1 + a\, p^{\alpha-1}]^2},\qquad a \geq 0,\quad c > 0, 
\quad 0 \leq \alpha \leq 1
\end{equation}
Time-domain relaxation calculus and creep compliance for power law  attenuation can be
found in \citet{HanMTAV}. Applications of this model in acoustic inversion and 
seismic inversion
can be found in \citet{RibodettiHanygaGJI} and \citet{FastTrack}. 
For $d = 1, 3$ the Green's function can be expressed in terms of an exponential
$$U^{(d)}(p,r) = F^{(d)}(p,r) \, \e^{-\Phi(p,r)}$$
where $F^{(d)}(p,r)$ is an algebraic function and the phase function 
\begin{equation}
\Phi(p,r) = \rho^{1/2}\, p\, r/Q(p)^{1/2} = \rho^{1/2}\, \left[p + a\, p^\alpha   \right]   \,r/c 
\end{equation}
The attenuation is non-negative:
$$\vert  \e^{-\Phi(-\ii \omega,r)} \vert = \e^{-\re \Phi(-\ii \omega,r)} = 
\e^{-a\, r^\alpha\, \omega^\alpha\, \cos(\upi\,\alpha/2)} \leq 1$$
and the Green's function is causal by a Paley-Wiener theorem (\citet{PaleyWiener}, 
Theorem~XII) because 
$$
\int_0^\infty \frac {\vert \ln \vert F^{(d)}(r,p)\,\e^{-\Phi(-\ii \omega,r)}\vert 
\vert}{1 + \omega^2} \,\dd \omega 
\approx \frac{a \, \rho^{1/2} r \cos(\upi \alpha/2)}{c} 
\int_0^\infty \frac{\omega^\alpha}{1 + \omega^2} \, \dd \omega < \infty
$$

The above results can be extended to $d = 2$ by using the asymptotic formula
$K_0(z) \sim \sqrt{\upi/2 z} \, \exp(-z)$.

\begin{theorem} \label{thm:QCBFpowerlaw} 
\begin{enumerate}[(i)]
\item The function $Q^{1/2}$ is a CBF for $0 \leq \alpha \leq 1$.
\item The function $Q$ is a CBF for $1/2 \leq \alpha \leq 1$. 
\end{enumerate}
\end{theorem}
\begin{proof}
We note that for positive real $p$ both $Q(p)$ and $Q(p)^{1/2}$ are real and non-negative.
If $0 \leq \arg p < \upi$ then
$$
\Psi := \arg \left[ 1 + a\, p^{\alpha-1}  \right] \geq \arg \left[ a\, p^{\alpha-1}  \right]
= -(1 - \alpha) \, \arg p \geq -(1 -\alpha) \, \upi
$$
and $\Psi \leq 0$.
Thus
$$0 \leq \arg \left[ 1 + a\, p^{\alpha-1}  \right]^{-1} \leq (1 - \alpha)\, \upi$$
By Theorem~\ref{thm:analCBF} $Q^{1/2} \in \CBF$. 
However
$$0 \leq \arg \left[ 1 + a\, p^{\alpha-1}  \right]^{-2}  \leq 2 (1 -\alpha)\, \upi \leq \upi$$ holds only
for $\alpha \geq 1/2$. Hence $Q$ is a CBF for $1/2 \leq \alpha \leq 1$.
\end{proof}

\noindent The results of the previous subsection imply that the Green's function 
is non-negative for $0 \leq \alpha \leq 1$. 
\begin{corollary}
The relaxation modulus $G$ is CM if and only if $1/2 \leq \alpha \leq 1$.
\end{corollary}
\begin{proof}
$G$ is CM if and only if $Q$ is a CBF. Hence $G$ is CM for
$1/2 \leq \alpha \leq 1$. 

We shall prove by direct calculus that $G$ is not
CM and therefore $Q$ is not a CBF for $0 \leq \alpha < 1/2$.
The relaxation modulus is given by the inverse Laplace transform
\begin{equation}
G(t) = \frac{c^2}{2 \upi \ii} \int_\mathcal{B} \e^{p t} 
\frac{1}{p\, [1 + a\, p^{\alpha-1}]^2}\, \dd p 
\end{equation}
where $\mathcal{B}$ denotes the Bromwich contour running parallel to the imaginary axis 
in the right-half complex plane. For $a \geq 0$ and 
$0 \leq \alpha \leq 1$ 
the denominator does not have zeros on the Riemann sheet $-\upi < \arg p < \upi$
and all the singularities of the integrand are situated on the cut along
the negative real axis. The integrand does not have a pole at 0 nor on the 
cut. For $t > 0$ the Bromwich contour can be replaced by a half-circle
at infinity in the half-plane $\re p \leq 0$ and the Hankel loop $\mathcal{H}$ encircling the
cut from $-\infty$ below the cut to zero and back to $-\infty$ above the cut
(Fig.~\ref{fig:Contour}). The contribution of the half-circle vanishes,  hence 
\begin{multline*}
G(t) = \frac{c^2}{2 \upi \ii} \int_\mathcal{H} \e^{p t}
\frac{1}{p\, \left[1 + a\, p^{\alpha-1}\right]^2}\, \dd p = \\
\frac{c^2}{\upi} \int_0^\infty \e^{-R t}
\im \frac{1}{R\,[1 + a\, (R\, \exp(\ii \upi))^{\alpha-1}]^2}\, \dd R = 
\\ \frac{c^2}{\upi} \int_0^\infty \e^{-R t} \frac{2 a \, R^\alpha \, \sin(\upi \alpha) 
- a^2\, R^{2 \alpha - 1} \, \sin(\upi (1 - 2 \alpha))}{\left[R + a^2 \, R^{2 \alpha -1} 
- 2 a \, R^\alpha \, \cos(\upi \alpha)   \right]^2} \dd R
\end{multline*}
The integrand is non-negative for $\alpha \geq 1/2$ and changes sign for $\alpha < 1/2$.
By Bernstein's theorem the relaxation modulus $G$ is CM for $\alpha \geq 1/2$ and is not 
CM for $\alpha < 1/2$.
\begin{figure}
\begin{center}
\includegraphics[width=0.25\textwidth]{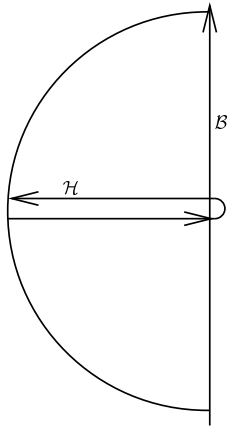} 
\end{center}
\caption{Contour deformation.} \label{fig:Contour}
\end{figure}
Fig.~\ref{fig:G} shows that $G$ is neither monotone nor non-negative 
for $\alpha < 1/2$.
\end{proof}
\begin{figure}
\begin{center}
\includegraphics[width=0.75\textwidth]{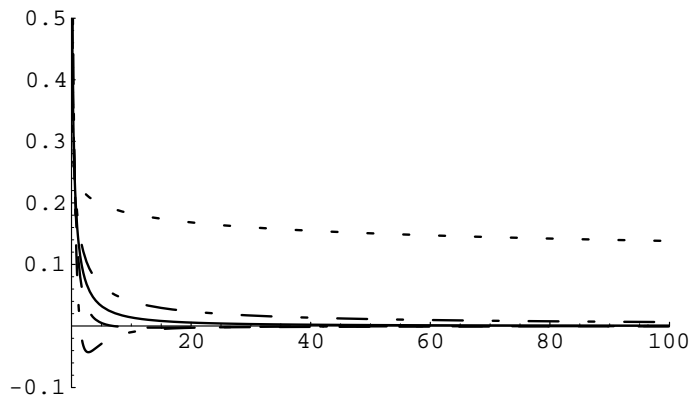}
\end{center}
\caption{The relaxation modulus for the power-law attenuation. The exponent 
values are 
$\alpha = 0.3, 0.4, 0.5, 0.6, 0.9$,
from bottom to top.} \label{fig:G}
\end{figure}
It follows that $Q$ is not a CBF if $\alpha < 1/2$. Hence we have extended positivity to 
the case when $Q \not\in \CBF$ but $Q^{1/2} \in \CBF$. 

\section{Concluding remarks.}
\label{sec:concluding}

In an elastic medium and in the case of Newtonian viscosity a non-negative source term excites 
a non-negative displacement pulse. This positivity property is shared by hereditary viscoelastic 
media with non-negative relaxation spectra and is limited to 1,2 and 3-dimensional problems.
The concept of positivity can be extended to systems of partial differential equations 
(vector fields).
For general matrix-valued pseudo-differential operators (anisotropic viscoelasticity) positivity
holds in a non-local sense. In isotropic viscoelasticity it holds for longitudinal and transverse 
waves separately.

The example in Sec.~\ref{subsec:example} shows that positive monotonicity of $G$ is not a necessary
condition for positivity. In Sec.~\ref{subsec:extended} positivity is shown to hold for partial 
differential equations of a more general type for dimension $\leq 3$.

\bibliography{mrabbrev,ownnew12,mathnew12}

\appendix
\section{A lemma.}

\begin{lemma} \label{lem:app}
If $f$ is a BF, then the function $f(x)/x$ is CM.
\end{lemma}
\begin{proof}
Let
$$ f(x) = a + b\, x + \int_{]0,\infty[} \left[ 1 - \e^{-x\,y}  \right] \, m(\dd y)$$
with $a, b \geq 0$ and a positive Radon measure $m$. 
By Bernstein's theorem the function
$$\frac{1 - \e^{-x}}{x} = \int_0^1 \e^{-x y} \, \dd y$$
is CM. Therefore  
$$\frac{f(x)}{x} = \frac{a}{x} + b + \int_{]0,\infty[} \frac{1 - \e^{-x y}}{x} \, m(\dd y)$$
is the superposition of CM functions with non-negative weights. By 
Theorem~\ref{thm:integralsuperposition} the function $f(x)/x$ is CM.  
\end{proof}

\section{Some properties of Stieltjes functions and CBFs.}

\begin{definition}
A real function $f$ on $\mathbb{R}_+$ is a Stieltjes function if it has the integral 
representation 
$$f(x) = a + \int_{]0,\infty[\;} \frac{\mu(\dd y)}{x + y}$$
with $a \geq 0$ and a positive Radon measure $\mu$ satisfying \eqref{eq:ineq}.
\end{definition}
The Stieltjes functions form a cone $\mathfrak{S}$:\\
if $a, b \geq 0$ and $f, g \in \mathfrak{S}$  then $a \, f + b \, g \in \mathfrak{S}$.

Every Stieltjes function is a CM function.
\begin{theorem}\label{thm:CBFStieltjes}
A function $f$ is a CBF if and only if $f(x)/x$ is a Stieltjes function.
\end{theorem}
Theorem~\ref{thm:CBFStieltjes} follows from the integral representation 
\eqref{eq:CBFintegral}.
Compare Theorem~\ref{thm:CBFStieltjes} with Corollary~\ref{corr:BF}.

We shall also need the following theorem:
\begin{theorem} \label{thm:logconvCBF}
If $f, g \in \CBF$ and $0 \leq \alpha \leq 1$ then $h := f^\alpha \, g^{1-\alpha} \in \CBF$.
\end{theorem}
\begin{proof}
If $f, g \in \CBF$ then $f$ and $g$ map $\mathbb{R}_+$ into $\overline{\mathbb{R}_+}$
and $\mathbb{C}_+$ into itself. Consequently $h$ maps $\mathbb{R}_+$ into 
$\overline{\mathbb{R}_+}$ and 
$$0 \leq \arg h(z) \leq \alpha\, \upi  + (1-\alpha) \, \upi = \upi$$
for $z \in \mathbb{C}_+$. Hence, by Theorem~\ref{thm:analCBF}, $h \in \CBF$.
\end{proof}

\end{document}